\documentclass[a4paper, 12pt]{article}

\usepackage[utf8x]{inputenc}
\usepackage{amsmath}
\usepackage{graphicx}
\usepackage[mathscr]{eucal}
\usepackage{mathrsfs} 
\usepackage{textcomp}
\usepackage{amsmath}
\usepackage{amssymb}
\usepackage{amsmath, amsthm}
\usepackage{latexsym}
\usepackage{amsfonts}
\usepackage{graphicx}
\usepackage{graphics}
\usepackage{mathrsfs}
\usepackage{color}
\usepackage{cases}
\usepackage{hyperref}

\theoremstyle{plain}
\newtheorem{theorem}{Theorem}

\title{Towards the Exact Simulation Using Hyperbolic Brownian Motion}
\author{Yuuki Ida and Yuri Imamura}
\date{}

\begin{document}
\maketitle
\begin{abstract}
In the present paper, 
an expansion of the transition density of Hyperbolic Brownian motion with drift is given, which is potentially useful for pricing and hedging of options under stochastic volatility models.
We work on a condition on the drift 
which dramatically simplifies the proof. 
\end{abstract}
\section{Introduction}

As is well recognized, 
``local stochastic volatility models"
can be reduced to 
Brownian motion with drift
thanks to Lamperti's transform. 
This is not the case 
when one works with 
stochastic volatility 
(henceforth SV) models 
where the stock price $ S $ and 
its instantaneous volatility $ V $ 
are modeled by a two-dimensional 
diffusion process.
One can not transform it into 
a two dimensional Brownian motion with drift in general. 

As is pointed in \cite{HL1},
however, 
most of existing stochastic volatility models are ``conformally equivalent''
to {\em hyperbolic} Brownian motion
(HBM for short)
instead; or in other words, 
many
SV diffusion processes
$ (S, V) $ can be transformed to 
HBM with drift
by a diffeomorphism.

In the present paper, 
we shall give an asymptotic expansion  
formula of the transition density
of HBM with drift
with respect to the so-called {\em McKean kernel}; 
density kernel. That is, the HBM \underline{without drift}. 
We claim that this formula 
can be used in  
numerical calculations for the
under SV models, although in this paper we will not go in depth in this direction. 

Our formula is in fact a parametrix one,
so along the line of 
Bally-Kohatsu \cite{BK}'s idea, 
we give an exact simulation interpretation of the parametrix
formula\footnote{Here the term ``exact'' is used because it is not an approximation, 
but the equality. It may be also referred to as ``unbiased''
since it is only simulate the expectation of a functional of $(S_t,V_t)$.}.

The present paper is organized as follows. In section \ref{HBM},
we briefly recall some basic facts about HBM. 
In section \ref{HBMwD}, we introduce 
a drift to the HBM, 
and describe its transition density 
by using as parametrix a HBM (Theorem \ref{mainII}. 
In section \ref{ESI}, 
we give an interpretation of the formula given in Theorem \ref{mainII}
that it gives a description of an {\em exact simulation}. 

In the present paper we restrict ourselves 
to 1) working on a simple situation
given by \eqref{HBMmu}; no drift 
in the volatility, and \eqref{condi-mu},  
which reduce the computational complexity of the proof dramatically. 
Further, 2) we omit the description of
how SV models can be transformed to HBM in this paper.
The main aim of the present paper 
is then to show that the condition \eqref{HBMmu} simplifies the proof quite a lot.


\section{Hyperbolic Brownian Motions}\label{HBM}
In this section, we recall 
basic facts about 
hyperbolic Brownian motions.

Let $ n \geq 2 $ and 
\begin{equation*}
\mathbb{H}^n
:= \{ z = (x,y) = 
(x^1, \cdots, x^{n-1}, y); x \in \mathbb{
R}^{n-1}, y > 0 \}, 
\end{equation*}
the upper half space
in $ \mathbb{R}^n $, endowed 
with the Poincar\'e metric
\footnote{A metric, at each point, is a bi-linear form on the tangent space, 
or equivalently, an element of the tensor product of the cotangent space. The convention $ (dx)^2 $
should then be understood as 
$ dx \otimes dx $, and so on. }
\begin{equation*}
ds^2 = y^{-2}((dx)^2 +(dy)^2).
\end{equation*}
The Riemannian volume element is given by
$ dv = y^{-2} dxdy $ and the distance 
$ d_{\mathbb{H}^n} (z,z') $ for $ z= (x,y) $, 
$ z' = (x',y') \in \mathbb{H}^n $ 
is given by 
\begin{equation}\label{equation2}
\cosh ( d_{\mathbb{H}^n}(z,z'))
= \frac{d_{\mathbb{R}^{n-1}}(x,x')^2 + y^2 + (y')^2}{
2 y y'}.
\end{equation}
The Laplace-Beltrami operator is 
\begin{equation*}
\Delta_n
:=y^2 \sum_{i=1}^{n-1}
\frac{\partial^2}{\partial x_i^2}
+ y^2 \frac{\partial^2}{\partial y^2}
- (n-2) y \frac{\partial}{\partial y}.
\end{equation*}
We denote by $ q_n (t,z,z') $ 
the heat kernel with respect to 
the volume element
$ dv $ of the semigroup 
generated by $ \Delta_n /2 $; that is to say, 
\begin{equation*}
\partial_t q_n = \frac{1}{2} \Delta_n q_n,
\end{equation*}
and
\begin{equation*}
\lim_{t \to 0} \int_{\mathbb{H}^n}
q_n (t,z',(x,y)) f(x,y) y^{-2} dxdy =
f(z')
\end{equation*}
for any bounded continuous function $ f $.
In other words, 
\begin{equation}\label{trd}
\mathbb{P} ( (X_t,Y_t) \in dxdy |(X_0,Y_0) =z')
= q_n (t,z',(x,y)) y^{-2} dxdy,
\end{equation}
where $ (X_t, Y_t) $ is the solution to 
the following stochastic differential equation:
\begin{equation}\label{HBMn}
\begin{split}
dX^i_t &= Y_t dW^i_t, \,\, i=1,\cdots, n-1,\\
dY_t &= Y_t dW^n_t,
\end{split}
\end{equation}
where $ W^1, \cdots, W^n $ are mutually independent Brownian motion defined on a probability space $(\Omega, \mathcal{F},\mathbb{P})$. 
The diffusion $ (X,Y) $ is the one
associated with the semigroup $ \Delta_n/2 $.

The following formulas for $ q_n $
are known (see e.g.\cite{Da} and \cite{MYII}):
\begin{theorem}\label{HKHBM}
The heat kernel with respect to 
the volume form has the following 
explicit expressions. \\
i) (McKean's kernel) 
In the case of $ n =2 $;
\begin{equation*}
q_2(t,z,z')=:p_2(t,r)=\frac{\sqrt{2}e^{-t/8}}{(2\pi t)^{3/2}}\int^\infty_r\frac{be^{-b^2/2t}}{(\cosh(b)-\cosh(r))^{1/2}}db.
\end{equation*}
ii) (Milson's formula) 
For $ n \geq 2 $, we have the following recursive relation;
\begin{equation*}
q_{n+2}(t,z,z')=:p_{n+2}(t,r)=-\frac{e^{-nt/2}}{2\pi \sinh(r)}\frac{\partial}{\partial r}p_n(t,r).
\end{equation*}
(iii) (Gruet's formula \cite{Gru}) 
For every $n\geq 2, t>0, z, z' \in \mathbb{H}^n$, it holds that
\begin{equation*}
\begin{split}
q_n(t, z,z')&=
p_n(t,r)\\
&=\frac{e^{-(n-1)^2t/8}}{\pi (2\pi)^{n/2}t^{1/2}}\Gamma\left(\frac{n+1}{2} \right)\int^\infty_0 \frac{e^{(\pi^2-b^2)/2t}\sinh(b)\sin(\pi b/t)}{(\cosh(b)+\cosh(r))^{(n+1)/2}}db,
\end{split}
\end{equation*}
where $r=d_{\mathbb{H}^2}(z,z')$.
\end{theorem}

\section{HBM with drift, and its parametrix}\label{HBMwD}

We consider the following 
stochastic differential equation:
\begin{equation}\label{HBMmu}
\begin{split}
dX_t &= Y_t dW^1_t 
+ \mu (X_t, Y_t) \,dt \\
dY_t &= Y_t dW^2_t, \\
(X_0,Y_0) &= (x,y) = z,
\end{split}
\end{equation}
where $ (x,y) =z \in \mathbb{H}^2 $,
$ \mu:\mathbb{H}^2 \to \mathbb{R} $
be a Lipschitz function,
bounded in $ x $ and
\begin{equation}\label{condi-mu}
|\mu (x,y)| \leq K_0 |y|,
\quad (x,y) \in \mathbb{H}^2
\end{equation}
with some positive constant $ K_0 $. 
The unique strong solution to (\ref{HBMmu})
exists, and will be denoted by $ (X^\mu, Y^\mu) =: Z^\mu $, 
while the 2-dimensional HBM 
given by \eqref{HBMn} with $ n=2 $ 
will be denoted by $ (X^0, Y^0)
=:Z^0 $.  

Put
\begin{equation*}
\begin{split}
\theta (t,z, z') 
&:= \mu (x,y) 
\frac{\partial}{\partial x} 
\log q_2 (t,(x, y), (x', y')) \\
&= \mu (x,y) \frac{\frac{\partial}{\partial x} q_2 (t,(x, y), (x', y'))}{q_2 (t,(x, y), (x', y'))}, \\
& \qquad t >0, z, z' \in \mathbb{H}^2. 
\end{split}
\end{equation*}
For $ t > 0 $ and each $ n $, 
let
$$ \Delta_n (t) := \{ 
(u_1,u_2, \cdots, u_{n}) \in [0,t]^{n} : u_1 < \cdots< u_{n} \}.
$$

The following is the main theorem of the present paper:

\begin{theorem}\label{mainII}
(i) We have that 
\begin{equation}\label{theta}
|\theta (t,z,z')| \leq \frac{3K_0}{2}
\end{equation}
and therefore for each $ n \geq 2 $, $ t > 0 $ and $ (s_1, \cdots, s_{n-1}) \in \Delta_{n-1} (t) $,
the random variable 
$ \prod_{i=1}^n \theta 
(s_i-s_{i-1}, Z^0_{s_{i-1}}, Z^0_{s_i}) $, 
where $ s_0=0 $ and $ s_n =t $, 
is in $ L^\infty (\mathbb{P})$
 and  
\begin{equation*}
\mathbb{E} [\prod_{i=1}^n \theta 
(s_i-s_{i-1}, Z^0_{s_{i-1}}, Z^0_{s_i})|Z^0_t=z' ] \in L^\infty (\Delta_{n-1} (t))
\end{equation*}
for each $ t > 0 $ and $ z, z' \in \mathbb{H}^2 $. 

(ii) Set 
\begin{equation}\label{h1def}
h_1(t,z,z')=\mu (x,y)\frac{\partial }{\partial x}q_2(t,z,z') (y')^{-2}. 
\end{equation}
and
\begin{equation*}
h_n (t,z,z') := \int_{\Delta_{n-1} (t)}
\mathbb{E} [\prod_{i=1}^n \theta 
(s_i-s_{i-1}, Z^0_{s_{i-1}}, Z^0_{s_i})|Z^0_t=z' ] q_2 (t,z,z') ds_1 \cdots ds_{n-1}
\end{equation*}
for $n\geq 2$. 
Then, the series
$ \sum_{n=1}^N h_n (t, z,z')$ is absolutely convergent 
as $ N \to \infty $ uniformly 
in $ (t,z, z') $ on every compact set.

(iii) The transition density 
of $ Z^{\mu} $
is given by
\begin{equation*}
\frac{q_2(t,z,z')}{(y')^2} +\int_{\mathbb{H}^2}\int^t_0
\frac{q_2(t-s,z,z'')}{(y'')^2} \Phi(s,z'',z')dsdz'',
\end{equation*}
where $\Phi(t,z,z')=\sum^\infty_{n=1}h_n(t,z,z')$. 
\end{theorem}

\begin{proof}
\if1
We set $z = z_0$, $z_i = (x_i,y_i)$ 
and $z'= z_n=(x_n, y_n)$. Then
\begin{equation}\label{sqre-int-eq}
\begin{split}
&\mathbb{E}
\left[| F_n (t; s_1, \cdots, s_{n-1}, z,Z^0_{s_1},
\cdots,Z^0_{s_{n-1}}, Z^0_t)|^2 \right]
\\&= 
\int_{\mathbb{H}^n} 
\prod_{j=1}^{n}
|\theta (s_j-s_{j-1}, z_{j-1}, z_j)|^2 
\prod_{i=1}^{n} q_2 (s_i-s_{i-1}, z_{i-1}, z_{i})y_i^{-2}
dz_1 dz_2 \cdots dz_{n-1}dz_n
\\&= 
\int_{\mathbb{H}^n} 
\prod_{i=1}^{n}
\left(|
\mu (x_{i-1}, y_{i-1})|^2 
\left(
\frac{\partial}{\partial x_{i-1}}q_2(s_i-s_{i-1}, z_{i-1}, z_i)
\right)^2\right.
\\
&\quad \times\left. \left(\frac{q_2 (s_i-s_{i-1}, z_{i-1}, z_{i})}{q_2 (s_i-s_{i-1}, z_{i-1}, z_{i})^2}\frac{1}{y_i^2}\right)\right)
dz_1 dz_2 \cdots dz_{n-1}dz,
\end{split}
\end{equation}
where we conventionally set 
$ s_n := t $ and $ s_0 := 0 $.
\fi

Since $q_n(t,z,z')=p_n(t,r(z,z'))$, 
we have that 
\begin{equation*}
\begin{split}
\frac{\partial}{\partial x} 
q_2(t,(x,y), (x',y')) 
&=\frac{\partial}{\partial x}
p_2(t,r((x,y), (x',y')))\\
&=\frac{\partial r}{\partial x}\frac{\partial p_2}{\partial r}
(t,r((x, y), (x', y')))
\\&= \frac{x - x'}
{y y'\sinh (r)}
\left(-e^{t}2 \pi \sinh (r) 
p_4(t,r((x, y), 
(x', y')))\right)\\
&=\frac{x - x'}
{y y'}
\left(-e^{t}(2 \pi )
p_4(t,r((x, y), 
(x', y')))\right)
\end{split}
\end{equation*}
by (ii) of Theorem \ref{HKHBM}.
Also, (iii) of Theorem \ref{HKHBM} tells us that 
\if1
\begin{equation*}
\begin{split}
\frac{p_4(t,r)}{p_2(t,r)} 
&=
\frac{e^{(-3^2 + 1)\frac{t}{8}}}{ (2\pi)^{\frac{(4-2)}{2}}}\frac{\Gamma(\frac{5}{2})}{\Gamma(\frac{3}{2})}
\frac{\int_0^{\infty}
\frac{e^{\frac{\pi^2 -b^2}{2t}}\sinh (b) \sin (\frac{\pi b }{t})}{(\cosh (b) + \cosh (r))^{\frac{5}{2}}}db}
{\int_0^{\infty}
\frac{e^{\frac{\pi^2 -b^2}{2t}}\sinh (b) \sin (\frac{\pi b }{t})}{(\cosh (b) + \cosh (r))^{\frac{3}{2}}}db}
\\&\leq
\frac{e^{-t}}{(2\pi)}\frac{3}{2}
\frac{1}{(1 + \cosh (r))}.
\end{split}
\end{equation*}
\fi
\begin{equation*}
\begin{split}
e^t(2\pi)p_4(t,r)&=\frac{e^{-t/8}}{\pi(2\pi)t^{1/2}}\Gamma\left(\frac{5}{2}\right)\int_0^{\infty}
\frac{e^{(\pi^2 -b^2)/2t}\sinh (b) \sin (\pi b/t)}{(\cosh (b) + \cosh (r))^{5/2}}db\\
&= \frac{3}{2}\frac{e^{-t/8}}{\pi(2\pi)t^{1/2}}\Gamma\left(\frac{3}{2} \right)
\int_0^{\infty}
\frac{e^{(\pi^2 -b^2)/2t}\sinh (b) \sin (\pi b/t)}{(\cosh(b)+\cosh(r))(\cosh (b) + \cosh (r))^{3/2}}db\\
&\leq\frac{3}{2}\frac{1}{1+\cosh(r)}p_2(t,r)
\end{split}
\end{equation*}
since $\cosh(x)\geq 1$ for all $x$.
Therefore, we see that 
\begin{equation*}\label{equation11}
\begin{split}
|\theta (t, (x, y), (x', y'))|
& \leq 
|\mu (x,y)| \frac{|\frac{\partial}{\partial x }
q_2(t,(x, y), (x', y')) |}
{q_2( t,(x, y), (x', y'))}
\\& \leq 
\frac{3K_0}{2}
\frac{ |y||x - x'|}
{ y y' (1 + \cosh (r(z,z')))}.
\end{split}
\end{equation*}

Here, we have used \eqref{condi-mu} in the last inequality. 
By \eqref{equation2},
\begin{equation*}
\begin{split}
& \frac{ |y||x - x'|}
{ y y' (1 + \cosh (r(z,z')))}\\
& =  \frac{ |y||x - x'|}
{ y y' (1 + \frac{|x-x'|^2 + y^2 + (y')^2}{
2 y y'})}= 
\frac{ 2 |y||x - x'|}
{|x-x'|^2 + |y + y'|^2} \\
& \leq \frac{ |y|}
{|y + y'|}
\leq 1.
\end{split}
\end{equation*}
Thus we obtained \eqref{theta}.
Here in the last line 
we have used the 
following elementary inequality: 
\begin{equation*}
|x-x'|^2 + |y + y'|^2
\geq 2|x-x'||y + y'|.
\end{equation*}

Let us consider (ii). 
By \eqref{theta}, we have that for $n$ bigger than 2,
\begin{equation}\label{bdofhn}
\begin{split}
h_n (t,z,z') 
&\leq \frac{q_2 (t,z,z')}{(y')^2}
\int_{\Delta_{n-1}(t)}\mathbb{E}[ \left(\frac{3}{2}K_0\right)^{n} | Z^0_t =z'] ds_1 \cdots ds_{n-1} \\
&= \left(\frac{3}{2}K_0\right)^{n}
\frac{q_2 (t,z,z')}{(y')^2} 
\int_{\Delta_{n-1}(t)}ds_1 \cdots ds_{n-1} \\
&= \left(\frac{3}{2}K_0\right)^{n}
\frac{q_2 (t,z,z')}{(y')^2} \frac{t^{n-1}}{(n-1)!}.
\end{split}
\end{equation}
Here we have used
\begin{equation*}
\mathbb{E}\left[1| Z_t^0=z' \right]=\frac{q_2(t,z,z')}{(y')^2}.
\end{equation*}
Hence we have
\begin{equation*}\label{equation9}
\begin{split}
\sum_{n=1}^{\infty} |h_n (t, z,z')| 
& \leq \frac{q_2 (t,z,z')}{(y')^2} \sum_{n=1}^{\infty}  \left(\frac{3}{2}K_0\right)^{n}
\frac{t^{n-1}}{(n-1)!} 
\\& = \frac{3}{2}K_0 \frac{q_2 (t,z,z')}{(y')^2}
\sum_{n=0}^{\infty}  \left(\frac{3}{2}K_0 t \right)^{n}
\frac{1}
{n!}
\\&= \frac{3}{2}K_0 \frac{q_2 (t,z,z')}{(y')^2}
e^{\frac{3}{2}K_0 t},
\end{split}
\end{equation*}
which complete the proof of (ii).

Finally, we shall prove (iii).  Since 
\begin{equation*}
\begin{split}
h_n (t,z,z') = \int_{\mathbb{H}^2}\int^t_0 h_1 (t-s,z,z'') h_{n-1} (s,z'',z')dsdz'',
\end{split}
\end{equation*}
we see that 
the sum $\sum^\infty_{n=1}h_n(t,z,z')
=:\Phi(t,z,z')$ satisfies 
\begin{equation}\label{equation21}
\begin{split}
\Phi(t,z,z')=
h_1 (t,z,z') + \int_{\mathbb{H}^2}\int^t_0 h_1 (t-s,z,z'') \Phi(s,z'',z')dsdz''.
\end{split}
\end{equation}
Note that since we have, by \eqref{equation9},
\begin{equation*}
\begin{split}
|\Phi(t,z,z')|&=|\sum^\infty_{n=1}h_n(t,z,z')|\\
&\leq \sum^\infty_{n=1}|h_n(t,z,z')| \leq \frac{3}{2}K_0 \frac{q_2 (t,z,z')}{(y')^2}
e^{\frac{3}{2}K_0 t},
\end{split}
\end{equation*}
we see that $\Phi$ is integrable:
\begin{equation*}
\begin{split}
\int^T_0\int_{\mathbb{H}^2}|\Phi(t,z,z')|dz'dt&\leq \frac{3}{2}K_0\int^T_0\int_{\mathbb{H}^2}\frac{q_2 (t,z,z')}{(y')^2}
e^{\frac{3}{2}K_0 t}dz'dt\\
& \leq \frac{3}{2}K_0e^{\frac{3}{2}K_0T}\int^T_0\int_{\mathbb{H}^2}\frac{q_2 (t,z,z')}{(y')^2}dz'dt=\frac{3}{2}K_0Te^{\frac{3}{2}K_0T}<\infty.
\end{split}
\end{equation*}

We know that 
\begin{equation*}
\left(\frac{1}{2}\Delta_2-\partial_t\right)q_2(t,z,z') =0,
\end{equation*}
and 
\begin{equation*}
\begin{split}
& \left(\frac{1}{2}\Delta_2-\partial_t\right)\int_{\mathbb{H}^2}\int^t_0 
\frac{q_2(t-s,z,z'')}{(y'')^{2}} \Phi(s,z'',z')dsdz'' \\
& \hspace{3cm} = -\Phi(t,z,z')
\end{split}
\end{equation*}
by Feynman-Kac formula (see e.g. \cite[Theorem 7.6]{KS}).  
Therefore, we have that 
\begin{equation*}
\begin{split}
&\left(\frac{1}{2}\Delta_2+\mu\frac{\partial}{\partial x_1}-\partial_t\right)p_2(t,z,z')\\
&=\left(\frac{1}{2}\Delta_2+\mu\frac{\partial}{\partial x_1}-\partial_t\right)\left(\frac{q_2(t-s,z,z')}{(y')^{2}}+\int_{\mathbb{H}^2}\int^t_0 \frac{q_2(t-s,z,z'')}{(y'')^{2}} \Phi(s,z'',z')dsdz''\right)\\
&=\mu\frac{\partial q_2}{\partial x_1} \frac{1}{(y'^2)}
+\int_{\mathbb{H}^2}\int^t_0
\frac{\mu}{(y'')^2} \frac{\partial q_2}{\partial x_1} (t-s,z,z'')\Phi(s,z'',z')dsdz''-\Phi(t,z,z'),
\end{split}
\end{equation*}
which is seen to be zero by \eqref{h1def} and \eqref{equation21}.

Clearly, 
the property that $ p_2 (t,z,z')
dz $ converges to $ \delta_{z'} (dz) $ is inherited from $ q_2 $.
\end{proof}

\if1 
\begin{equation*}
\begin{split}
\end{split}
\end{equation*}
\fi 

\section{Exact Simulation Interpretation}\label{ESI}

In the spirit of Bally-Kohatsu \cite{BK}, we give the following ``exact simulation interpretation'' to Theorem \ref{mainII}.

\begin{theorem}
Let $ S_i $, $ i=1\, \cdots $, 
are independent copies of
an exponentially distributed random variable
with mean $ 1 $, 
which are also independent of the Brownian motion 
$ (W^1, W^2) $. 
Let $ T_i := S_1 + \cdots + S_i $
and $ N_t := \sum_i 1_{\{ T_i \leq t \}} $, $ t> 0 $. Then,
for any bounded measurable $ f $, we have that
\if4
\begin{equation}
\begin{split}
\mathbb{E} [|\prod_{i=1}^{N_t} \theta 
(T_i - T_{i-1}, Z^0_{T_{i-1}}, Z^0_{T_i}) |^2] < \infty.
\end{split}
\end{equation}
and 
\fi
\begin{equation*}
\begin{split}
\mathbb{E} [ f(Z^\mu_t) ]
= e^t \mathbb{E} [\prod_{i=1}^{N_t} \theta 
(T_i - T_{i-1},  Z^0_{T_{i-1}}, Z^0_{T_i})
f(Z^0_t)].
\end{split}
\end{equation*}
\end{theorem}
Even though this is an almost direct corollary 
to Theorem \ref{mainII} and Bally-Kohatsu's general theory, 
we give a self-contained proof below.
\begin{proof}
First we claim that
for a positive measurable function 
$$ G \equiv G (s_1, \cdots s_{k+1}, z_1, \cdots, z_{k+1} ) , $$ 
we have that
\begin{equation}
\begin{split}\label{claim1}
&\mathbb{E}\left[1_{\{N_t=k\}} 
G (T_1, \cdots, T_{k+1}, 
Z^0_{T_1}, \cdots, Z^0_{T_k},Z^0_t ) \right]\\
& = \mathbb{E}\left[
\int_{\Delta_{k} (t) \times [t, \infty)}
G (s_1, \cdots, s_{k+1}, 
Z^0_{s_1}, \cdots, Z^0_{s_k},Z^0_t ) 
ds_1 \cdots ds_{k}e^{-s_{k+1}}ds_{k+1} \right].
\end{split}
\end{equation}

In fact, since 
\begin{align*}
&\mathbb{E}\left[1_{\{N_t=k\}} 
G (T_1, \cdots, T_{k+1}, 
Z^0_{T_1}, \cdots, Z^0_{T_k},Z^0_t ) \right]\\
&=\mathbb{E}\left[\mathbb{E}\left[
1_{\{T_1\leq t, \cdots, T_k\leq t, T_{k+1}>t\}}
G (T_1, \cdots T_{k+1}, 
Z^0_{T_1}, \cdots, Z^0_{T_k},Z^0_t )
|\mathcal{F}^Z\right]\right] \\
&=\mathbb{E}\left[
\int_{[0,t]^k \times (t,\infty)}
G (s_1, \cdots, s_{k+1}, 
Z^0_{s_1}, \cdots, Z^0_{s_k},Z^0_t )
\mathbb{P}(T_1\in ds_1,
 \cdots, T_{k+1}\in ds_{k+1} )\right],
\end{align*}
and since the joint density of 
$ T_1, \cdots, T_k $ 
is given by
\begin{align*}
& \mathbb{P} (T_1 \in ds_1, \cdots,
T_{k+1} \in ds_{k+1} )  \\
& = 1_{\{s_{k+1}> s_k >s_{k-1}>\cdots >s_1 > 0 \}}e^{-s_{k+1}} ds_1\cdots ds_{k+1}, 
\end{align*}
we have \eqref{claim1}.
\if3
\begin{align*}
&\mathbb{E}\left[
\int_{[0,t]^k \times [t,\infty)}
G (s_1, \cdots s_{k+1}, 
Z_{s_1}, \cdots, Z_{s_k},Z_t )
\mathbb{P}(T_1\in ds_1,
 \cdots, T_{k+1}\in ds_{k+1} )\right] \\
&=\mathbb{E}\left[
\int_{\Delta_{k} (t)\times [t,\infty)}
G (s_1, \cdots s_{k+1}, 
Z_{s_1}, \cdots, Z_{s_k},Z_t ) e^{-s_{k+1}}
ds_1 \cdots ds_{k+1}\right] 
\end{align*}
\fi

In particular, 
if G is independent to $s_{k+1}$,
we have the following reduction:
\begin{equation}
\begin{split}\label{reducedC1}
&\mathbb{E}\left[1_{\{N_t=k\}} 
G (T_1, \cdots T_{k}, 
Z^0_{T_1}, \cdots, Z^0_{T_k},Z^0_t ) \right]\\
& =e^{-t} \mathbb{E}\left[
\int_{\Delta_{k} (t)}
G (s_1, \cdots s_{k}, 
Z^0_{s_1}, \cdots, Z^0_{s_k},Z^0_t ) 
ds_1 \cdots ds_{k} \right].
\end{split}
\end{equation}

We note that 
we can apply \eqref{reducedC1}
to
\begin{align*}
G_+ (s_1, \cdots s_{k}, z_1, \cdots, z_{k+1} ) = \left(\prod^k_{i=1}\theta (s_i-s_{i-1}, z_{{i-1}},z_i)f(z_{k+1})\right)_+, \\
G_- (s_1, \cdots s_{k}, z_1, \cdots, z_{k+1} ) = \left(\prod^k_{i=1}\theta (s_i-s_{i-1}, z_{{i-1}},z_i)f(z_{k+1})\right)_-,
\end{align*}
and so we have 
\begin{equation}\label{thetaf}
\begin{split}
&\mathbb{E}\left[\prod^k_{i=1}1_{\{N_t=k\}}\theta (T_i-T_{i-1}, Z^0_{T_{i-1}},Z^0_{T_i}) f(Z^0_t) \right]
\\
& = 
e^{-t} 
\mathbb{E}\left[
\int_{\Delta_{k} (t)}
\prod^k_{i=1}\theta (s_i-s_{i-1}, Z^0_{s_{i-1}},Z^0_{s_i}) f(Z_t)
ds_1 \cdots ds_{k} \right].
\end{split}
\end{equation}
Since we know from (i) of Theorem
\ref{mainII} that 
$ \prod^k_{i=1}\theta (s_i-s_{i-1}, Z^0_{s_{i-1}},Z^0_{s_i}) \in L^\infty (\mathbb{P}) $, we see that
$ \prod^k_{i=1}\theta (s_i-s_{i-1}, Z^0_{s_{i-1}},Z^0_{s_i}) f(Z^0_t) $ 
is in $ L^1 (\mathbb{P}) $
by the requirement that 
$ f (Z^0_t) \in L^1 (\mathbb{P}) $.
Therefore, the right-hand-side of 
\eqref{thetaf} is equal to 
\begin{align*}
& e^{-t}
\int_{\Delta_{k} (t)}
\mathbb{E}\left[ \prod^k_{i=1}\theta (s_i-s_{i-1}, Z^0_{s_{i-1}},Z^0_{s_i}) f(Z^0_t) \right]
ds_1 \cdots ds_{k}.
\end{align*}
Noting that
\begin{align*}
& \mathbb{E}\left[ \prod^k_{i=1}\theta (s_i-s_{i-1}, Z^0_{s_{i-1}},Z^0_{s_i}) f(Z^0_t) \right]
\\
&= 
\int_{(\mathbb{H}^2)^{k+1}}
\prod^k_{i=1} h_1 (s_i -s_{i-1},z_{i-1},z_i )
f(z') \frac{q_2 (t-s_k, z_k, z')}{(y')^2} \,dz_1 \cdots dz_k dz', \\ 
\end{align*}
we obtain that
\begin{equation}\label{thetaz}
\begin{split}
&\int_{\Delta_{k} (t)}
\mathbb{E}\left[ \prod^k_{i=1}\theta (s_i-s_{i-1}, Z^0_{s_{i-1}},Z^0_{s_i}) f(Z^0_t) \right]
ds_1 \cdots ds_{k} \\
&=
\int_{\mathbb{H}^2} \left( \int_{(\mathbb{H}^2 \times [0,t])^k}
\prod^k_{i=1} h_1 (s_i -s_{i-1},z_{i-1},z_i ) 
\frac{q_2 (t-s_k, z_{k},z')}{(y')^2} ds_{i}dz_i \right)
f(z') \,dz'\\
&= \int_{\mathbb{H}^2} \left( \int_{\mathbb{H}^2 \times [0,t]}
h_k (s_k,z'',z' ) 
\frac{q_2 (t-s_k, z,z'')}{(y'')^2} ds_{k} dz'' \right)
f(z') \,dz',
\end{split}
\end{equation}
which is bounded by 
\begin{equation*}
\begin{split}
& \left(\frac{3}{2}K_0\right)^{k}\frac{t^{k-1}}{(k-1)!}
\int_{\mathbb{H}^2}
\frac{q_2 (t,z,z')}{(y')^2} |f(z')| \,dz'\\
& = \left(\frac{3}{2}K_0\right)^{k}\frac{t^{k-1}}{(k-1)!}
\mathbb{E} [|f(Z^0_t)|],
\end{split}
\end{equation*}
as we see from \eqref{bdofhn}.  
Therefore, 
we can change the order between 
the summation and the expectation in 
\begin{align*}
&\mathbb{E}\left[\prod^{N_t}_{i=1} \theta (T_i-T_{i-1}, Z^0_{T_{i-1}},Z^0_{T_i})f(Z^0_t)\right]\\
&=\mathbb{E}\left[f(Z^0_t)1_{\{N_t=0\}}+\sum^{\infty}_{k=1}\prod^k_{i=1}1_{\{N_t=k\}}\theta (T_i-T_{i-1}, Z^0_{T_{i-1}},Z^0_{T_i}) f(Z^0_t) \right]. \\
\end{align*}
On the other hand,
by \eqref{thetaz},
\begin{equation*}
\begin{split}
&\mathbb{E}\left[f(Z^0_t)1_{\{N_t=0\}}\right] +\sum^{\infty}_{k=1}
\mathbb{E} \left[ \prod^k_{i=1}
1_{\{N_t=k\}}\theta (T_i-T_{i-1}, Z^0_{T_{i-1}},Z^0_{T_i}) f(Z^0_t) \right]
\\
&= e^{-t} \int_{\mathbb{H}^2}
\frac{q_2(t,z,z')}{(y')^2} f(z') \,dz' \\
&+
e^{-t} \sum_{k=1}^\infty 
\int_{\mathbb{H}^2} f(z') dz'
\int_{\mathbb{H}^2}\int^t_0
\frac{q_2(t-s,z,z'')}{(y'')^2} h _k (s,z'',z')dsdz'' \\
&= e^{-t} \int_{\mathbb{H}^2}
\frac{q_2(t,z,z')}{(y')^2} f(z') \,dz' \\
&+
e^{-t} 
\int_{\mathbb{H}^2} f(z') dz'
\int_{\mathbb{H}^2}\int^t_0
\frac{q_2(t-s,z,z'')}{(y'')^2} \Phi (s,z'',z')dsdz'' \\
&= e^{-t} \mathbb{E} [ f(Z^{\mu}_t)],
\end{split}
\end{equation*}
where the last equality is valid by
(iii) of Theorem \ref{mainII}.

\if5
Since we know from (ii) of 
Theorem \ref{mainII} that 
$ \sum_n h_n $ is uniformly convergent, 

we can apply \eqref{claim1}
to obtain that 
\begin{equation}\label{exact01}
\begin{split}
& \mathbb{E}\left[\prod^{N_t}_{i=1} | \theta (T_i-T_{i-1}, Z_{T_{i-1}},Z_{T_i})f(Z_t) |^2 \right] \\
&=e^{-t}
\sum_{k=0}^\infty \mathbb{E}\left[ \int^{t}_0\int^{s_k}_0 \cdots \int^{s_2}_0\prod^{k}_{i=1}|\theta\left(s_i-s_{i-1},Z_{s_{i-1}}, Z_{s_i}\right)f(Z_t)|^2ds_1\cdots ds_k\right] \\
& = e^{-t}
\sum_{k=0}^\infty \int^{t}_0\int^{s_k}_0 \cdots \int^{s_2}_0
\mathbb{E} \left[| F_{k} (t; s_1, \cdots, s_{k}, z, Z^0_{s_1},
\cdots,Z^0_{s_{k-1}}, Z^0_{s_k})f(Z_t)|^2 \right] ds_1\cdots ds_k \\
&\leq e^{-t}\sum^{\infty}_{k=0}\int_0^t \mathbb{E}\left[K\frac{t^{k}}{k!Y_{s_k}}\right] ds_k
=e^{-t}\sum^{\infty}_{k=0}
\int_0^t \mathbb{E}\left[K\frac{t^k}{k!y_0}\exp\left(\frac{s_k}{2}-W^2_{s_k}\right)\right] ds_k \\
&=e^{-t}\sum^{\infty}_{k=0}K\frac{t^k}{k!y_0}\int^t_0e^{s_k}ds_k\leq K\sum^{\infty}_{k=0}\frac{t^{k+1}}{k!y_0}=K\frac{te^t}{y_0}
\end{split}
\end{equation}
which is apparently finite. 
Here, the last inequality 
is due to \eqref{bdofhn}. \\
Applying 
\begin{equation*}
G(s_1,\cdots, s_{k+1},z_1,\cdots,z_k,z'):=\prod^k_{i=1}\theta(s_i-s_{i-1},z_{i-1},z_i)f(z')
\end{equation*}
\fi
\end{proof}


\begin{thebibliography}{99}

\bibitem{BK}
Bally, K. and Kohatsu-Higa, A. (2015)
``A probabilistic interpretation of the parametrix method'',
{\em Ann. Appl. Probab.}, 
Volume 25, Number 6, 3095-3138. 


\bibitem{Da}
Davies, E.B. (1989)
{\em Heat Kernels and Spectral Theory}, Cambridge Univ.
Press.

\bibitem{Gru}
Gruet, J.-C. (1996) ``Semi-groupe du mouvement Brownien hyperbolique",
{\em Stochastics Stochastic Rep.}, 56, 53-61.


\bibitem{HL1}
Henry-Labord\`ere, P. (2005)
``A General Asymptotic Implied Volatility for Stochastic Volatility Models", 
arXiv:cond-mat/0504317

\bibitem{KS}
Karatzas, I and Shreve, S. E. (1991) {\em Brownian Motion and Stochastic Calculus}, Second Edition, Springer-Verlag

\bibitem{MYII}
Matsumoto, H and and Yor, M. (2005)
``Exponential functionals of Brownian motion, II: Some related diffusion processes", {\em Probab. Surveys}
Volume 2, 348-384.

\end{thebibliography}
\end{document}